\documentclass[letter]{ieice}
\usepackage[pdftex]{graphicx,xcolor}
\usepackage[fleqn]{amsmath}
\usepackage{amsthm}
\usepackage{newtxtext}
\usepackage[varg]{newtxmath}
\usepackage{tikz} 
\usepackage{enumitem}
\usepackage{booktabs}

\field{D}

\title{Complexity of the Existence of Constrained Secure Equilibria in
  Multi-Player Games}

\titlenote{%
  This work was supported by JSPS KAKENHI Grant Numbers 23K24824 and 24K14901.}

\authorlist{%
 \authorentry{Hiroki Mizuno}{n}{nagoya}%
 \authorentry[takata.yoshiaki@kochi-tech.ac.jp]{Yoshiaki TAKATA}{m}{kochi}%
 \MembershipNumber{9618813}
 \authorentry[seki@sqlab.jp]{Hiroyuki Seki}{f}{nagoya}%
 \MembershipNumber{8206514}
}
\affiliate[nagoya]{The authors are with the
 Graduate School of Informatics,
 Nagoya University,
 Furo-cho, Chikusa, Nagoya 464-8601, Japan.}
\affiliate[kochi]{The author is with the
 School of Informatics,
 Kochi University of Technology,
 Tosayamada, Kami City, Kochi 782-8502, Japan.}

\received{20xx}{1}{1}
\revised{20xx}{1}{1}

\newtheorem{pro}{Proposition}
\newtheorem{thm}{Theorem}

\theoremstyle{definition}
\newtheorem{dfn}{Definition}
\newtheorem{ex}{Example}
\theoremstyle{remark}

\newcommand{\llr}[1]{\langle\!\langle{#1}\rangle\!\rangle}
\newcommand{\calA}{\mathcal{A}}
\newcommand{\calG}{\mathcal{G}}
\newcommand{\Inf}{\mathit{inf}}
\newcommand{\Pay}{\mathit{Pay}}

\newcommand{\Out}{\mathit{Out}}
\newcommand{\Buchi}{\relax
  \ifmmode\text{\textit{B\"{u}chi}}\else B\"{u}chi\fi}
\newcommand{\coBuchi}{\relax
  \ifmmode\text{\textit{co-B\"{u}chi}}\else co-B\"{u}chi\fi}

\newcommand{\Parity}{\mathit{Parity}}
\newcommand{\Rabin}{\mathit{Rabin}}
\newcommand{\Streett}{\mathit{Streett}}
\newcommand{\Muller}{\mathit{Muller}}
\newcommand{\Nat}{\mathbb{N}}
\newcommand{\subarena}[2]{{#1\mathbin{\upharpoonright}#2}}
\newcommand{\SE}{\mathit{SE}}

\usetikzlibrary{shapes.geometric,positioning,calc}

\begin{document}
\maketitle
\begin{summary}
We consider a multi-player non-zero-sum turn-based game
(abbreviated as multi-player game) on a finite directed graph.
A secure equilibrium (SE) is a strategy profile in which
no player has the incentive to deviate from the strategy
because no player can increase her own payoff or lower the payoff of another player.
SE is a promising refinement of Nash equilibrium in which
a player does not care the payoff of another player.
In this paper, we discuss the decidability and complexity of the problem
of deciding whether a secure equilibrium with constraints
(a payoff profile specifying which players must win) exists
for a given multi-player game.
%
\end{summary}
\begin{keywords}
multi-player game, secure equilibrium, winning region
\end{keywords}


\section{Introduction}
Games can be seen as mathematical models
for analyzing the behaviors of multiple objects
competing or cooperating with one another.
One application of games
is the \emph{reactive synthesis}~\cite{BCJ18},
in which given a specification of the input and output of a system,
a program satisfying that specification is automatically synthesized.

In a typical setting of the reactive synthesis,
one player (called the system)
in a game represents a reactive program to be synthesized,
and the other players (called the environments)
model the behavior of uncontrollable entities.
Each uncontrollable entity is modeled as a player who
has her own \emph{objective} and
acts in a way that best fulfills the objective.
A solution of the reactive synthesis problem is
a locally-optimal tuple of the behaviors (or \emph{strategies})
of the players,
called an \emph{equilibrium},
where no player has incentive to change her strategy.

\emph{Nash equilibrium} (NE) is a well-known
definition of equilibrium,
where
no player can improve her payoff
(1 if her objective is satisfied and 0 otherwise)
by changing her strategy alone.
For broad classes of games, the existence of NE is guaranteed.
However, NE does not consider the payoffs of the other players.
Therefore,
when we adopt an NE as a solution of the reactive synthesis problem,
we cannot guarantee that an uncontrollable entity $p$
never chooses a behavior that interferes with the system
without changing $p$'s payoff.
To overcome the weakness, refinements of NE such as
secure equilibrium \cite{CHJ04}
have been proposed.

\emph{Secure equilibrium} (SE) is the one where
no player can improve her payoff or
decrease the payoff of any other player without decreasing her payoff
by changing her strategy alone.
Chatterjee et al.\ \cite{CHJ04} showed that there always exists
an SE from every state of a given two-player game.
Bruy\`{e}re et al.\ \cite{BMR14} extended the notion of SE to two-player weighted games.
They also proved that \emph{the constrained SE existence problem},
which is the problem to decide whether
there exists an SE that yields a given tuple of payoffs, is decidable
for several payoff measures.
Usually, the aim of a reactive synthesis is to find an equilibrium
where the system wins.
The constrained existence problem is a more general problem where
one can specify which players must win under an equilibrium.
Pril et al.\ \cite{PFK14} showed that
there always exists an SE from every state of a given multi-player (non-weighted) game.
However, the constrained SE existence problem was not considered in \cite{PFK14}.

In this paper, we study the decidability and complexity of
the constrained SE existence problem in multi-player games.
The results are summarized in Table~\ref{table:comp}.
The classes of objectives such as \Buchi\ and \coBuchi\ are
described in Section~\ref{subsec:objective}.

\begin{table}[tb]
  \centering
  \caption{The complexity of the constrained SE existence problem.}
  \label{table:comp}
  \newcommand{\results}[2]{%
    \begin{tabular}{c}#1 \\[-1pt] \scriptsize (#2)\end{tabular}}
  \begin{tabular}{cccc}
  \toprule
    & \multicolumn{3}{c}{Objectives} \\
  \cmidrule(lr){2-4}
    & \Buchi & \coBuchi & Parity \\
  \midrule
    Upper bound
    & \results{co-NP}{Th.~\ref{thm:upper-buchi}}
    & \results{NP}{Th.~\ref{thm:upper-cobuchi}}
    & \results{PSPACE}{Th.~\ref{thm:upper-muller}} \\
    Lower bound
    & \multicolumn{2}{c}{\results{P-hard}{Th.~\ref{thm:lower-buchi}}}
    & \results{NP-hard, co-NP-hard}{%
        Th.~\ref{thm:lower-parity-np} \& \ref{thm:lower-parity-conp}} \\
  \bottomrule
  \end{tabular}
\end{table}


\section{Definitions}\label{sec:def}


Let $\Nat=\{0, 1, 2, \ldots \}$ be the set of all natural numbers including~$0$.
For $i,j\in\Nat$,
let $[i,j]=\{k\in\Nat \mid i\le k\le j\}$.
For a set $S$, let
$S^*$ and $S^\omega$ denote the sets of
all finite and infinite sequences
on~$S$, respectively.
For an infinite sequence $\rho\in S^\omega$,
let $\Inf(\rho)$ denote the set of elements
appearing infinitely often in~$\rho$.

\subsection{Game Arena}
\label{subsec:arena}
\begin{dfn}
  An $n$-player \emph{game arena} is a triple
  $\calA = (S, (S_i)_{i\in[1,n]}, E)$,
  where $S$ is a nonempty set of \emph{states},
  $(S_i)_{i\in[1,n]}$ is a decomposition of~$S$,
  and $E\subseteq (S\times S)$ is a \emph{transition relation}
  between two states.
  $S_i$ represents the set of states controlled by player~$i$.
\end{dfn}

For a game arena $\calA = (S, (S_i)_{i\in[1,n]}, E)$
and a subset $S'\subseteq S$ of states,
let $\subarena{\calA}{S'}$ denote the sub-arena
obtained from $\calA$ by restricting the set of states to~$S'$;
i.e.,
$\subarena{\calA}{S'} =
 (S', (S_i\cap S')_{i\in[1,n]}, E\cap(S'\times S'))$.

A \emph{play} is an infinite sequence
$\rho = s_0s_1\ldots \in S^\omega$
of states satisfying $(s_k,s_{k+1})\in E$ for all $k\ge 0$.
In other words,
a play is an infinite sequence of states along with
the transition relation of a game arena.

\subsection{Winning Objectives}
\label{subsec:objective}

In this paper, we assume that the result
a player obtains from a play is either a winning or a losing.
Each player has her own winning condition over plays,
and we model the condition as a subset $\varphi$ of plays;
i.e., the player wins if the play belongs to the subset $\varphi$.

A \emph{winning objective} (or simply, \emph{objective})
$\varphi$ is a subset of $S^\omega$.
For a play $\rho$,
a player with an objective $\varphi$ \emph{wins}
if $\rho\in\varphi$ and
the player \emph{loses} if $\rho\notin\varphi$.

In this paper,
we consider the following classic and important
classes of objectives.
\begin{description}[font=\mdseries\itshape,nosep,topsep=\smallskipamount]
\item[\Buchi\ objective]
  is given by a subset $B\subseteq S$ as\\
  $\Buchi(B)=\{\,\rho\in S^{\omega} \mid \Inf(\rho)\cap B\neq \emptyset\,\}$.
\item[Co-\Buchi\ objective]
  is given by a subset $C\subseteq S$ as\\
  $\coBuchi(C)=\{\,\rho\in S^{\omega} \mid \Inf(\rho)\cap C = \emptyset\,\}$.
\item[Parity objective]
  is given by a coloring function $p:S\to\Nat$ as
  $\Parity(p)=\{\,\rho\in S^{\omega} \mid
   \min_{s\in\Inf(\rho)} p(s) \text{ is even}\,\}$.
\item[Streett objective]
  is given by a subset $T \subseteq 2^{S}\times 2^{S}$ as\\
  $\Streett(T)=\bigcap _{(F,G)\in T} \bigl(\coBuchi(F)\cup \Buchi(G)\bigr)$.
\item[Rabin objective]
  is given by a subset $R \subseteq 2^{S}\times 2^{S}$ as\\
  $\Rabin(R)=\bigcup _{(F,G)\in R} \bigl(\Buchi(F)\cap \coBuchi(G)\bigr)$.
\item[Muller objective]
  is given by a Boolean formula $\phi$ over $S$ as
  $\Muller(\phi)=\{\,\rho\in S^{\omega} \mid \Inf(\rho)\models\phi\,\}$,
  where
  $\Inf(\rho)\models\phi$ means that $\phi$ is evaluated to $1$
  under the truth assignment $\theta: S\rightarrow \{0,1\}$
  such that $\theta(s)=1 \Leftrightarrow s\in\Inf(\rho)$.
\end{description}

Among the above classes of objectives,
the class of Muller objectives is the most expressive,
and any objective in the other classes can be translated
in polynomial time
into an equivalent Muller objective.

For an objective $\varphi$,
let $\neg\varphi = S^\omega\setminus\varphi$.
By definition, we can obtain
the following proposition.
\begin{pro}\label{buchi-pro}
The following equations hold for
$B_1,B_2\subseteq S$ and
$p:S\to\Nat$.
\begin{align*}
  \neg\Buchi(B_1) &= \coBuchi(B_1), \\
  \Buchi(B_1) \cup \Buchi(B_2) &= \Buchi(B_1 \cup B_2), \\
  \coBuchi(B_1)\cap\coBuchi(B_2) &= \coBuchi(B_1 \cup B_2), \\
  \neg\Parity(p) &= \Parity(p+1). \qquad\!\qed
\end{align*}
\end{pro}

A \emph{game} is a pair $\calG =(\calA,\Phi)$ of
a game arena $\calA$
and an \emph{objective profile}
$\Phi = (\varphi_1,\ldots,\varphi_n)$,
which is the tuple of objectives of players $1$ to~$n$.
If the objectives of the players are complementary,
i.e.,
$\varphi_i\cap\varphi_j = \emptyset$
for any distinct $i,j\in[1,n]$ and
$\bigcup_{i\in[1,n]}\varphi_i = S^\omega$,
then we say the game is \emph{zero-sum}.
Otherwise, the game is \emph{non-zero-sum}.
This paper mainly concerns non-zero-sum games.

\subsection{Strategies}\label{subsec:strategy}

A \emph{strategy} $\sigma_i$ of player~$i\in[1,n]$
is a function
$\sigma_i:S^*S_i\to S$, which
maps a finite sequence $hs\in S^*S_i$
ending with a state $s\in S_i$ controlled by $i$
into a succeeding state $s'$
satisfying $(s,s')\in E$.

For a subset $P\subseteq[1,n]$ of players,
a \emph{strategy profile of $P$} is a tuple
$\Sigma_P = (\sigma_i)_{i\in P}$ of strategies.
We call a strategy profile
$\Sigma = (\sigma_1,\ldots,\sigma_n)$
of all players simply a \emph{strategy profile}.

For player~$i$'s strategy $\sigma_i$ and a state $s_0$,
we define the set $\Out_{s_0}(\sigma_i)$ of plays
starting with $s_0$ and consistent with $\sigma_i$
as follows.
\begin{equation*}
  \Out_{s_0}(\sigma_i) = \{
    \begin{aligned}[t]
      &s_0s_1\ldots\in S^{\omega} \mid
      s_{k+1}=\sigma_i(s_0s_1\ldots s_k)
      \text{ for}\\
      &\quad \text{all } k\ge0 \text{ such that }
      s_k \in S_i \}.
    \end{aligned}
\end{equation*}
Moreover,
for a strategy profile $\Sigma_P = (\sigma_i)_{i\in P}$
of a subset $P$ of players, we define
$\Out_{s_0}(\Sigma_P) = \bigcap_{i \in P} \Out_{s_0}(\sigma_i)$.
Note that
for a strategy profile $\Sigma$ (of all players),
$\Out_{s_0}(\Sigma)$ is a singleton.
In the following,
we consider $\Out_{s_0}(\Sigma)$ represents a play
instead of a singleton of a play.

Given an objective profile
$\Phi = (\varphi_1,\ldots,\varphi_n)$,
a strategy profile $\Sigma$,
and a starting state $s$,
the \emph{payoff profile}
$\Pay_s^\Sigma(\Phi) = (v_1,\ldots,v_n)\in \{0,1\}^n$
the players obtain when starting a play with $s$ and
all players follow $\Sigma$ is determined;
for each player $i\in[1,n]$,
$v_i=1$ if $Out_s(\Sigma)\in\varphi_i$ and
$v_i=0$ otherwise.
In a zero-sum game,
exactly one component of a payoff profile should be~$1$,
while in a non-zero-sum game,
any number of components of a payoff profile may be~$1$.

Given a game arena $\calA$, an objective $\varphi$, and
a subset $P\subseteq[1,n]$ of players,
we define the \emph{winning region}
$\llr{P}_{\calA}(\varphi)$
as the set of starting states from which
players in $P$ can cooperate to make $\varphi$
be satisfied.
Formally,
$\llr{P}_{\calA}(\varphi)$ is defined as follows.
\begin{equation*}
  \llr{P}_{\calA}(\varphi) = \{
      s\in S \mid
    \begin{aligned}[t]
      &\Out_s(\Sigma_P)\subseteq\varphi
      \text{ for some strategy}\\
      &\text{profile }\Sigma_P
      \text{ of }P \}.
    \end{aligned}
\end{equation*}
We write $\llr{P}_{\calA}(\varphi)$ as
$\llr{P}(\varphi)$ if $\calA$ is clear from the context.
Moreover,
we sometimes write
$\llr{\{p_1,\ldots,p_k\}}(\varphi)$
as $\llr{p_1,\ldots,p_k}(\varphi)$
for simplicity.

The following proposition on winning regions is known.
\begin{pro}[{\cite{GH82}}]
\label{determinacy}
  Consider an $n$-player game and
  let $\varphi$ be an objective,
  $I\subseteq[1,n]$ a subset of players,
  and $J = [1,n]\setminus I$.
  Then, $\llr{I}(\varphi) = S\setminus \llr{J}(\neg\varphi)$.
  \qed
\end{pro}

\subsection{Secure Equilibria}\label{subsec:se}

For player~$i\in[1,n]$,
we define the \emph{preference order} $\prec_i$
over payoff profiles as follows:
For payoff profiles $v=(v_1,\ldots,v_n)$ and
$v'=(v'_1,\ldots,v'_n)$,
%
\begin{equation*}
  v \prec_i v' \Leftrightarrow (v_i < v'_i) \lor
  \bigl(
    \begin{aligned}[t]
      (v_i = v'_i) &\land (\forall j.\ v_j \geq v'_j) \\
      {} &\land (\exists j.\ v_j > v'_j)\bigr).
    \end{aligned}
\end{equation*}
That is,
we consider player~$i$ prefers
a payoff profile $v'$ rather than $v$
if changing $v$ into $v'$ increases $i$'s payoff
or keeps $i$'s payoff the same,
does not increase
the other players' payoffs, and
decreases at least one player's payoff.

We define the \emph{secure equilibrium} below.
In the definition, we use a notation
for derived strategy profiles:
For a strategy profile
$\Sigma = (\sigma_1,\ldots,\sigma_n)$ and
a strategy $\sigma'_i$ of player~$i\in[1,n]$,
let
$(\sigma'_i,\Sigma^{-i})$ denote
the strategy profile
$(\sigma_1,\ldots,\sigma_{i-1},
  \sigma'_i,\sigma_{i+1},\ldots,\sigma_n)$.

\begin{dfn}
  Let $\calG = (\calA,\Phi)$ be a game.
  If a strategy profile
  $\Sigma$ 
  satisfies the following condition at a state~$s$,
  then we call $\Sigma$ a \emph{secure equilibrium (SE)}
  at~$s$.
\begin{quote}
   For all $i\in[1,n]$,
   there is no strategy $\sigma'_i$ of player~$i$
   that satisfies
   $\Pay_s^\Sigma(\Phi) \prec_i\Pay_s^{(\sigma'_i,\Sigma^{-i})}(\Phi)$.
\end{quote}  
\end{dfn}

For a game $\calG$ and a payoff profile $v$,
let $\SE_v$ ($\subseteq S$) denote the subset
of states at which there exists a secure equilibrium
that yields the payoff profile~$v$.


\section{Characterization of SE in Multi-Player Games}

This paper is aimed at investigating
the complexity of the following decision problem.
%
\begin{dfn}\label{def:existence-se}
  The \emph{constrained secure equilibrium existence problem} is the decision problem defined by the following instance and question:
  \begin{description}[font=\bfseries,nosep,topsep=\smallskipamount]
    \item[Instance]
      An $n$-player game $\calG = (\calA,\Phi)$,
      a starting state $s$,
      and a vector $v\in\{0,1\}^n$ called \emph{constraint}.
    \item[Question]
      Is there an SE $\Sigma$ at~$s$ satisfying
      $\Pay_s^\Sigma(\Phi) = v$?
  \end{description}
\end{dfn}

The question of the above problem can be rephrased as
``$s\in\SE_v$?\,''
We approach this problem by
representing the subset $\SE_v$ of states
using winning regions.
Let a given game be $\calG=(\calA,\Phi)$ where
$\Phi=(\varphi_1,\ldots,\varphi_n)$
and a given constraint be $v=(v_1,\ldots,v_n)$.
Let $I=[1,n]$ be the set of players and
subsets $W, L$ of players be
$W=\{\,i \mid v_i = 1\,\}$ and
$L=\{\,i \mid v_i = 0\,\}$, respectively.
Let $\varphi_W=\bigcap_{i\in W}\varphi_i$
and $\varphi_L=\bigcup_{i\in L}\varphi_i$.
The following algorithm computes the set~$\SE_v$.
\begin{enumerate}
  \item Compute
    $A_{v} =
    \begin{aligned}[t]
    \textstyle\bigcap_{w\in W}
      &\llr{I\setminus\{w\}} \bigl(
        \varphi_W \cup \varphi_L \cup \neg\varphi_w
      \bigr) \cap {} \\
    \textstyle\bigcap_{l\in L}
      &\llr{I\setminus\{l\}} \bigl(
        (\varphi_W \cup \varphi_L) \cap \neg\varphi_l
      \bigr).
    \end{aligned}$
  \item Compute
    $\llr{I}_{\subarena{\calA}{A_v}}
      \bigl(\varphi_W \cap \neg\varphi_L\bigr)$.
\end{enumerate}

$A_v$ is the subset of states from which
for each player~$i$,
the other players in $I\setminus\{i\}$ can cooperate
to prevent the payoff profile
becoming better than $v$ for~$i$.
For example,
consider a player~$w$ who is specified as a winner in
the constraint $v=(v_1,\ldots,v_n)$;
i.e., $v_w = 1$.
Player~$w$ prefers a payoff profile $v'$ rather than $v$
if $w$ is still a winner in $v'$,
all players who are specified as losers in $v$
are still losers in $v'$, and
at least one player who is specified as a winner in $v$
becomes a loser in~$v'$.
In other words, it is better for $w$ that
an objective
$\neg\varphi_W \cap \neg\varphi_L \cap \varphi_w$
is satisfied
(i.e.,
at least one player in $W$ loses,
no player in $L$ wins, and
$w$ wins).
To form a secure equilibrium,
the other players in $I\setminus\{w\}$
have to force the negation
$\varphi_W \cup \varphi_L \cup \neg\varphi_w$
of the above objective
to be satisfied
regardless of $w$'s strategy.
The first line of the definition of $A_v$
represents the subset of states
such that the players in $I\setminus\{w\}$
have a strategy profile $\Sigma_{I\setminus\{w\}}$ of them
(called a \emph{retaliation strategy profile})
that forces
$\varphi_W \cup \varphi_L \cup \neg\varphi_w$
to be satisfied
when a play starts from a state in the subset.
The second line of the definition of $A_v$ is similar.

The winning region
$\llr{I}_{\subarena{\calA}{A_v}}\bigl(\varphi_W \cap \neg\varphi_L\bigr)$
in Step~2 represents
the set of states
such that
there is a strategy profile $\Sigma$
(called a \emph{cooperation strategy profile})
that keeps a play inside of $A_v$ and
yields the payoff profile equal to $v$
when the play starts from a state in the region.

The following theorem shows
the winning region computed in Step~2 equals~$\SE_v$.
\begin{thm}\label{thm:correctness}
  Let $\calG = (\calA,\Phi)$ where
  $\Phi=(\varphi_1,\ldots,\varphi_n)$
  and
  $v=(v_1,\ldots,v_n)$ be a given game and constraint,
  respectively.
  Let $I=[1,n]$,
  $W=\{\,i \mid v_i = 1\,\}$,
  $L=\{\,i \mid v_i = 0\,\}$,
  $\varphi_W=\bigcap_{i\in W}\varphi_i$, and
  $\varphi_L=\bigcup_{i\in L}\varphi_i$.
  Then, equation
  $\SE_v = \llr{I}_{\subarena{\calA}{A_v}}\bigl(
    \varphi_W \cap \neg\varphi_L \bigr)$ holds,
  where
  \begin{align*}
    A_{v} ={}
      \textstyle \bigcap_{w\in W} \llr{I\setminus\{w\}}\bigl(
        \varphi_W \cup \varphi_L \cup \neg\varphi_w
      \bigr)& \cap{} \\
      \textstyle \bigcap_{l \in L} \llr{I\setminus\{l\}}\bigl(
        (\varphi_W \cup \varphi_L) \cap \neg\varphi_l
      \bigr)&.
  \end{align*}
\end{thm}
\begin{proof}
  First we show $\llr{I}_{\subarena{\calA}{A_v}}
     \bigl(\varphi_W \cap \neg\varphi_L\bigr)
     \subseteq \SE_v$.
  Suppose $s\in \llr{I}_{\subarena{\calA}{A_v}}
    \bigl(\varphi_W \cap \neg\varphi_L\bigr)$.
  This means that there is a cooperation strategy profile
  $\Sigma=(\sigma_1,\ldots,\sigma_n)$
  that keeps a play inside of $A_v$ and makes
  $\varphi_W \cap \neg\varphi_L$ be satisfied
  when the play starts from~$s$.
  Then, we construct a strategy profile
  $\Sigma'=(\sigma'_1,\ldots,\sigma'_n)$
  as follows.
  Each $\sigma'_i$ behaves as the same as $\sigma_i$
  as long as all the other players behave
  as the same as
  $(\sigma_1,\ldots,\sigma_{i-1},\sigma_{i+1},\ldots,\sigma_n)$.
  If some player~$i$ deviates from~$\Sigma$,
  i.e., at some point of a play
  player~$i$ chooses a state different from the one
  $\sigma_i$ chooses,
  then the other players' strategy profile
  $(\sigma'_1,\ldots,\sigma'_{i-1},\sigma'_{i+1},\ldots,\sigma'_n)$ switches its behavior
  to the one that makes
  $\varphi_W \cup \varphi_L \cup \neg\varphi_i$
  (resp.\ $(\varphi_W \cup \varphi_L) \cap \neg\varphi_i$)
  be satisfied when $i\in W$ (resp.\ $i\in L$).
  The outcome play of $\Sigma'$ is the same as
  $\Sigma$ (i.e.\ $\Out_s(\Sigma')=\Out_s(\Sigma)$)
  because every player behaves as the same as $\Sigma$
  as long as the other players behave so.
  Hence, $\Sigma'$ yields the payoff profile equal to~$v$.
  If some player~$i$ deviates from $\Sigma'$, then
  she cannot obtain a payoff profile better than $v$
  with respect to~$\prec_i$.
  Therefore, $\Sigma'$ is an SE
  that yields the payoff profile equal to~$v$.

  Next we show $\SE_v\subseteq
  \llr{I}_{\subarena{\calA}{A_v}}
    \bigl(\varphi_W \cap \neg\varphi_L\bigr)$
  by showing its contraposition.
  The complement
  $S\setminus \llr{I}_{\subarena{\calA}{A_v}}
    \bigl(\varphi_W \cap \neg\varphi_L\bigr)$
  of the right-hand side of the inclusion
  equals
  $\bigl(A_v\setminus\llr{I}_{\subarena{\calA}{A_v}}
   \bigl(\varphi_W \cap \neg\varphi_L\bigr)\bigr) \cup
   \bigl(S\setminus A_v\bigr)$.
  By Proposition~\ref{determinacy},
  $A_v\setminus\llr{I}_{\subarena{\calA}{A_v}}
     \bigl(\varphi_W \cap \neg\varphi_L\bigr)
  = \llr{\emptyset}_{\subarena{\calA}{A_v}}
     \bigl(\neg\varphi_W \cup \varphi_L\bigr)$ and
  $S\setminus A_v=\bigcup_{w \in W} \llr{\{w\}}
  \bigl(\neg\varphi_W \cap \neg\varphi_L \cap \varphi_w \bigr) \cup
  \bigcup_{l\in L}\llr{\{l\}}
  \bigl((\neg\varphi_W \cap \neg\varphi_L) \cup \varphi_l \bigr)$.
  By the latter equation,
  for any starting state $s\in S\setminus A_v$,
  there exists some player~$i$ who
  has a strategy that yields better payoff profile than $v$
  with respect to~$\prec_i$
  regardless of the strategy profile of the other players.
  (If $s\in \llr{\{w\}}\bigl(\neg\varphi_W\cap\neg\varphi_L\cap\varphi_w\bigr)$
    for some $w\in W$, then the player~$w$ has such a strategy.
    It is similar if
    $s\in \llr{\{l\}}\bigl((\neg\varphi_W\cap\neg\varphi_L)\cup\varphi_l\bigr)$
    for some $l\in L$.)
  This means that any strategy profile that
  yields the payoff profile equal to~$v$ and
  eventually visits a state $s'\in S\setminus A_v$
  is never an SE, because
  there is at least one player who
  can obtain better payoff profile than $v$
  by changing her strategy
  to the one that switches its behavior after reaching~$s'$.

  Suppose $s\in S\setminus\llr{I}_{\subarena{\calA}{A_v}}
    \bigl(\varphi_W \cap \neg\varphi_L\bigr)$.
  Then, $s\in \llr{\emptyset}_{\subarena{\calA}{A_v}}
    \bigl(\neg\varphi_W \cup \varphi_L\bigr)$ or
  $s\in S\setminus A_v$ hold.
  If $s\in S\setminus A_v$, then
  as mentioned above,
  there is no SE that yields the payoff profile equal to~$v$
  when starting from~$s$.
  If $s \in \llr{\emptyset}_{\subarena{\calA}{A_v}}
    \bigl(\neg\varphi_W \cup \varphi_L \bigr)$,
  then for the starting state~$s$,
  every strategy profile that keeps
  a play inside of $A_v$ does not yield the payoff profile
  equal to~$v$.
  Moreover, as mentioned above, any strategy profile
  that yields the payoff profile equal to~$v$ and
  visits a state outside of $A_v$ is not an SE\@.
  Therefore, for the starting state $s$,
  there is no SE that yields the payoff profile equal to~$v$.
  Thus,
  $SE_v\subseteq \llr{I}_{\subarena{\calA}{A_v}}
    \bigl(\varphi_W \cap \neg\varphi_L\bigr)$.
\end{proof}

\begin{ex}
  Figure~\ref{fig:buchigame} shows the arena of
  a 3-player game $\calG=(\calA,(\varphi_1,\varphi_2,\varphi_3))$
  where each of $\varphi_1,\varphi_2,\varphi_3$ is a \Buchi\ objective.
  Each circle represents a state player~$1$ controls,
  i.e.\ a member of~$S_1$.
  A rectangle represents a state player~$2$ controls and
  a diamond represents a state player~$3$ controls.
  The objectives of players $1, 2, 3$ are
  $\Buchi(\{s_2,s_4\})$, $\Buchi(\{s_0,s_5\})$, and
  $\Buchi(\{s_1,s_3\})$, respectively.
  The vector attached to each of $s_3, s_4, s_5$ represents
  the payoff profile obtained when a play reaches the attached state.
  We consider the constrained secure equilibrium existence problem
  for $\calG$, starting state $s_0$, and constraint $(1,1,1)$.

\begin{figure}[tb]
\centering
\begin{tikzpicture}[node distance=12mm,->,>=latex]
  \tikzstyle{circle} =[shape=circle,    draw, inner sep=0pt,  minimum size=1.8em]
  \tikzstyle{diamond}=[shape=diamond,   draw, inner sep=0pt,  minimum size=2.3em]
  \tikzstyle{rect}   =[shape=rectangle, draw, inner ysep=0pt, minimum size=1.6em]
  \tikzstyle{every node} = [font=\small]
  \node[circle] (s0) at (0, 0) {$s_0$};
  \node[rect]   (s1) at (1, -1.73) {$s_1$};
  \node[diamond](s2) at (2, 0) {$s_2$};
  \node[circle] (s3) [left of=s0] {$s_3$};
  \node[circle] (s4) [right of=s1] {$s_4$};
  \node[circle] (s5) [right of=s2] {$s_5$};
  \path
    (s0) edge [bend right] (s1)
    (s1) edge (s0)
    (s1) edge [bend right] (s2)
    (s2) edge (s1)
    (s2) edge [bend right] (s0)
    (s0) edge (s2)
    (s0) edge (s3)
    (s1) edge (s4)
    (s2) edge (s5)
    (s3) edge [loop below] (s3)
    (s4) edge [loop below] (s4)
    (s5) edge [loop below] (s5)
  ;
%
  \tikzstyle{every node} = [font=\scriptsize]
  \node [at=(s3.west), left] {$(0{,}0{,}1)$};
  \node [at=(s4.east), right] {$(1{,}0{,}0)$};
  \node [at=(s5.east), right] {$(0{,}1{,}0)$};
\end{tikzpicture}
\caption{A 3-player game with \Buchi\ objectives.}
\label{fig:buchigame}
\end{figure}
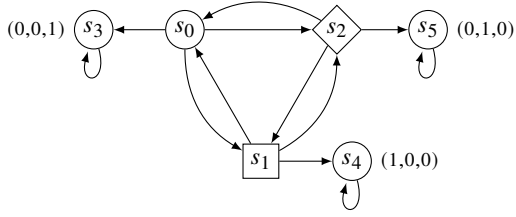

  As mentioned above,
  $A_{(1,1,1)}=
    \llr{1,2}
      \bigl(\varphi_{123} \cup \neg\varphi_3 \bigr) \cap
    \llr{2,3}
      \bigl(\varphi_{123} \cup \neg\varphi_1 \bigr) \cap
    \llr{1,3}
      \bigl(\varphi_{123} \cup \neg\varphi_2 \bigr)$ where
  $\varphi_{123}=\varphi_1\cap\varphi_2\cap\varphi_3$.
  Using an algorithm described in the proof of Theorem~\ref{thm:upper-buchi}
  in Section~\ref{upper},
  we obtain
  $\llr{1,2}\bigl(\varphi_{123}\cup \neg\varphi_3 \bigr)
     =\{s_0,s_1,s_2,s_4,s_5\}$,
  $\llr{2,3}\bigl(\varphi_{123}\cup \neg\varphi_1 \bigr)
     =\{s_0,s_1,s_2,s_3,s_5\}$,
  $\llr{1,3}\bigl(\varphi_{123}\cup \neg\varphi_2 \bigr)
     =\{s_0,s_1,s_2,s_3,s_4\}$,
  and thus $A_{(1,1,1)}=\{s_0,s_1,s_2\}$.
  Moreover, $\llr{I}_{\subarena{\calA}{A_{(1,1,1)}}}
    (\varphi_1 \cap \varphi_2 \cap \varphi_3)=\{s_0,s_1,s_2\}$,
  and therefore for the starting state~$s_0$,
  there exists an SE that yields the payoff profile equal to $(1,1,1)$.
  In fact, we can construct the following strategy profile
  $\Sigma = (\sigma_1,\sigma_2,\sigma_3)$
  that is an SE giving the payoff profile equal to~$(1,1,1)$:
\begin{align*}
  \sigma_1(hs_0) &= \begin{cases}
      s_1 & \text{if $h=\epsilon$ or } h=h's_2 \\
      s_3 & \text{if } h=h's_1
    \end{cases} \\
  \sigma_2(hs_1) &= \begin{cases}
      s_2 & \text{if $h=\epsilon$ or } h=h's_0 \\
      s_4 & \text{if } h=h's_2
    \end{cases} \\
  \sigma_3(hs_2) &= \begin{cases}
      s_0 & \text{if $h=\epsilon$ or } h=h's_1 \\
      s_5 & \text{if } h=h's_0
    \end{cases}
\end{align*}
  where $h,h'\in S^*$.
  As long as every player conforms to $\Sigma=(\sigma_1,\sigma_2,\sigma_3)$,
  the outcome play is $s_0s_1s_2s_0s_1s_2\ldots$\,,
  which yields payoff profile $(1,1,1)$.
  If some player deviates from $\Sigma$,
  then the other players retaliate against the deviating player
  by choosing $s_3$, $s_4$, or $s_5$ as the next state.
  For example, when player~$1$ chooses $s_2$ as the next state at~$s_0$,
  player~$3$ at $s_2$ selects $s_5$ as the next state.
  The payoff profile becomes $(0,1,0)$,
  which is not better than $(1,1,1)$ for player~$1$.
  The above $\Sigma=(\sigma_1,\sigma_2,\sigma_3)$ is one example of SEs
  and there is another SE whose outcome play is
  $s_0s_2s_1s_0s_2s_1\ldots$\,.
\end{ex}


\section{The Complexity of the Constrained SE Existence Problem}

In this section, we investigate the complexity of
the constrained secure equilibrium existence problem.

\subsection{Upper Bounds}
\label{upper}

\begin{thm}\label{thm:upper-buchi}
  For an $n$-player game
  where every player has a \Buchi\ objective,
  the constrained SE existence problem is in co-NP\@.
\end{thm}
\begin{proof}
  Let $\calG=(\calA,(\varphi_1,\ldots,\varphi_n))$
  and $v=(v_1,\ldots,v_n)$ be
  the given game and constraint, and let
  $I=[1,n]$, $W=\{\,i\mid v_i=1\,\}$, and $L=I\setminus W$.
  Let $\varphi_W=\bigcap_{i\in W}\varphi_i$ and
  $\varphi_L=\bigcup_{i\in L}\varphi_i$.
  Let $S$ be the set of states in~$\calA$.

  Consider the computation of
  $A_{v} = \bigcap_{w\in W} \llr {I \setminus \{w\}}
    \bigl(\varphi_W \allowbreak \cup \varphi_L \cup \neg\varphi_w \bigr) \cap
    \bigcap_{l \in L} \llr {I \setminus \{l\}}
    \bigl( (\varphi_W \cup \varphi_L) \cap \neg\varphi_l \bigr)$.
  By Proposition~\ref{buchi-pro},
  $\varphi_L$ can be represented by a single \Buchi\ objective
  and $\neg\varphi_w$ can be represented by a \coBuchi\ objective.
  Since
  $\varphi_W \cup \varphi_L \cup \neg\varphi_w
   = \bigcap_{i \in W}
     \bigl((\varphi_i\cup\varphi_L) \cup \neg\varphi_w\bigr)$ and
  $\varphi_i \cup \varphi_L$ for each $i\in W$ can also be
  represented by a single \Buchi\ objective,
  $\varphi_W \cup \varphi_L \cup \neg\varphi_w$ can be
  represented by a single Streett objective.
  In a similar way,
  $(\varphi_W \cup \varphi_L) \cap \neg\varphi_l
   = \bigl(\bigcap_{i\in W} (\varphi_i \cup \varphi_L)\bigr)
     \cap \neg\varphi_l
   = \bigl(\bigcap_{i\in W} ((\varphi_i\cup\varphi_L)\cup\coBuchi(S))\bigr)
     \cap (\neg\varphi_l\cup\Buchi(\emptyset))$
  can also be represented by a single Streett objective.
  Therefore,
  $A_v$ can be obtained by computing winning regions
  in 2-player zero-sum games with Streett objectives,
  which is in co-NP~\cite{EJ88}.

  The computation of
  $\llr{I}_{\subarena{\calA}{A_v}}
   \bigl(\varphi_W \cap \neg\varphi_L \bigr)$
  can be achieved by computing a winning region
  in a 1-player game with a Streett objective,
  which is in P~\cite{CH12}.
\end{proof}

\begin{thm}\label{thm:upper-cobuchi}
  For an $n$-player game
  where every player has a \coBuchi\ objective,
  the constrained SE existence problem is in NP\@.
\end{thm}
\begin{proof}
  Let $\calG=(\calA,(\varphi_1,\ldots,\varphi_n))$
  and $v=(v_1,\ldots,v_n)$ be the given game and constraint, and
  define
  $I$, $W$, $L$,
  $\varphi_W$, and
  $\varphi_L$ as the same as the proof of
  Theorem~\ref{thm:upper-buchi}.
  Let $S$ be the set of states in~$\calG$.

  By Proposition~\ref{buchi-pro},
  $\varphi_W$ can be represented by a single \coBuchi\ objective
  and $\neg\varphi_w$ can be represented by a \Buchi\ objective.
  Since
  $\varphi_W \cup \varphi_L \cup \neg\varphi_w
   = \varphi_W \cup \bigl(\bigcup_{i \in L} \varphi_i\bigr)
     \cup \neg\varphi_w
   = (\varphi_W \cap \Buchi(S)) \cup
     \bigl(\bigcup_{i\in L}(\varphi_i\cap\Buchi(S))\bigr)
     \cup (\neg\varphi_w\cap\coBuchi(\emptyset))$,
  $\varphi_W \cup \varphi_L \cup \neg\varphi_w$ can be
  represented by a single Rabin objective.
  In a similar way,
  $(\varphi_W \cup \varphi_L) \cap \neg\varphi_l
   = (\varphi_W\cap\neg\varphi_l)\cup
     \bigcup_{i\in L} (\varphi_i \cap\neg\varphi_l)$
  can also be represented by a single Rabin objective.
  Therefore,
  $A_v$ can be obtained by computing winning regions
  in 2-player zero-sum games with Rabin objectives,
  which is in NP~\cite{EJ88}.

  As the same as the proof of Theorem~\ref{thm:upper-buchi},
  the computation of
  $\llr{I}_{\subarena{\calA}{A_v}}
   \bigl(\varphi_W \cap \neg\varphi_L \bigr)$
  is in P~\cite{CH12}.
\end{proof}

\begin{thm}\label{thm:upper-muller}
  For an $n$-player game
  where every player has a Muller objective,
  the constrained SE existence problem is in PSPACE\@.
\end{thm}
\begin{proof}
  We can consider $A_v$ as an intersection of
  winning regions
  in 2-player zero-sum games with Muller objectives,
  and computing each of these winning regions is in PSPACE~\cite{HD05}.
  Hence, the computation of $A_v$ is also in PSPACE\@.
  The computation of
  $\llr{I}_{\subarena{\calA}{A_v}}
   \bigl(\varphi_W \cap \neg\varphi_L \bigr)$
  can be achieved by computing a winning region
  in a 1-player game with a Muller objective,
  which is in P~\cite{CH12}.
\end{proof}

\subsection{Lower Bounds}

\begin{thm}\label{thm:lower-buchi}
  For a 2-player game
  where both players have \Buchi\ objectives
  or have \coBuchi\ objectives,
  the constrained SE existence problem is P-hard.
\end{thm}
\begin{proof}
  An SE in a 2-player game that yields the payoff profile $(1,1)$
  is equivalent to a doomsday equilibrium (DE)~\cite{CDFR17}.
  Lemma~6 in~\cite{CDFR17} shows that the problem to decide whether
  a DE exists in a 2-player game
  with \Buchi\ or \coBuchi\ objectives is P-hard,
  and the same lower bound applies to
  the constrained SE existence problem.
\end{proof}

\begin{thm}\label{thm:lower-parity-np}
  For a 2-player game
  where both players have parity objectives,
  the constrained SE existence problem is NP-hard.
\end{thm}
\begin{proof}
  Lemma~8 in~\cite{CDFR17} shows that the problem to decide whether
  a DE exists in a 2-player game
  with parity objectives is NP-hard.
  In the same way as Theorem~\ref{thm:lower-buchi},
  the same lower bound applies to
  the constrained SE existence problem.
\end{proof}

\begin{thm}\label{thm:lower-parity-conp}
  For a 2-player game
  where both players have parity objectives,
  the constrained SE existence problem
  with constraint $(1,0)$ is co-NP-hard.
\end{thm}
\begin{proof}
  The problem to decide whether
  there exists a player having a winning strategy is
  co-NP-hard for
  a 2-player zero-sum game where
  one player has an objective that is the intersection
  of two parity objectives~\cite{CHP07}.
  We show that we can reduce this problem to
  the constrained SE existence problem
  for 2-player non-zero-sum games with parity objectives
  with constraint $(1,0)$.

  Suppose that we are given
  a 2-player zero-sum game $\calG =(\calA,(\varphi,\neg\varphi))$
  and a state $s$ in $\calA$ as a starting state.
  Assume that $\varphi$ is the intersection of
  two parity objectives $\varphi_1$ and~$\varphi_2$;
  i.e.,
  $\varphi = \varphi_1 \cap \varphi_2$.

  Consider a game $\calG' = (\calA,(\varphi_1,\neg\varphi_2))$
  having the same game arena as~$\calG$.
  Since $\neg\varphi_2$ is a parity objective
  by Proposition~\ref{buchi-pro},
  $\calG'$ is a 2-player non-zero-sum game with parity objectives.
%
  In~\cite{CHJ06},
  the SE existence problem for 2-player games
  and each of the four possible payoff profiles has been investigated.
  It shows
  $\SE_{(1,0)} = \llr{1}(\varphi_1 \cap \neg(\neg\varphi_2))$,
  which implies $\SE_{(1,0)} = \llr{1}(\varphi)$.
  This means that
  $s \in\SE_{(1,0)}$ in $\calG'$ if and only if
  player~$1$ has a winning strategy at state~$s$ in~$\calG$.
\end{proof}


\section{Conclusion}

Complexity of the constrained SE existence problem has been investigated.
An extension of equilibria including SE and DE as special cases and
the constrained existence problem for the extended equilibria will be reported in future.


\bibliographystyle{ieicetr}
\bibliography{ncrsp}


\end{document}